\documentclass[11pt,a4paper]{article}

\newif\ifdraft
\draftfalse

\usepackage[margin=1in]{geometry}
\usepackage{microtype}
\usepackage{lmodern}

\usepackage{amsmath,amsfonts,amssymb,amsthm}
\usepackage{thmtools,thm-restate}
\usepackage{mathtools}
\usepackage{nicefrac}

\usepackage{graphicx}
\usepackage{xcolor}
\usepackage{booktabs}
\usepackage{tabularx}
\usepackage{enumitem}

\usepackage{tikz}
\usetikzlibrary{calc}
\usepackage{pgfplots}
\pgfplotsset{compat=newest}
\usepackage{pgfplotstable}

\pgfplotsset{
  discard if not/.style 2 args={
    x filter/.code={
      \edef\tempa{\thisrow{#1}}
      \edef\tempb{#2}
      \ifx\tempa\tempb\else
        
      \fi
    }
  }
}
\pgfmathdeclarefunction{lg2}{1}{\pgfmathparse{ln(#1)/ln(2)}}
\pgfmathdeclarefunction{lg10}{1}{\pgfmathparse{ln(#1)/ln(10)}}

\usepackage[ruled,vlined]{algorithm2e}
\DontPrintSemicolon

\usepackage[square,numbers]{natbib}
\usepackage[colorlinks=true,
  linkcolor=red!40!black,
  citecolor=green!40!black,
  urlcolor=blue!40!black,
  filecolor=magenta!40!black,
  pdftitle={Approximate Counting of k-Paths},
  pdfdisplaydoctitle=true
]{hyperref}
\usepackage[nameinlink,noabbrev,sort&compress,capitalize]{cleveref}

\ifdraft
  \usepackage{showlabels}
  \usepackage[backgroundcolor=gray!10,textsize=footnotesize]{todonotes}
\else
  \newcommand{\todo}[2][]{}
\fi

\newtheorem{theorem}{Theorem}
\newtheorem{lemma}[theorem]{Lemma}

\newtheorem*{claim}{Claim}

\theoremstyle{definition}

\crefname{rrule}{Rule}{Rules}

\numberwithin{equation}{section}
\crefname{theorem}{theorem}{theorems}
\Crefname{theorem}{Theorem}{Theorems}
\crefname{lemma}{lemma}{lemmas}
\Crefname{lemma}{Lemma}{Lemmas}
\crefname{corollary}{corollary}{corollaries}
\Crefname{corollary}{Corollary}{Corollaries}

\title{Approximate Counting of $k$-Paths in $O^*(2^k)$ Time}
\author{Tomohiro Koana\\
\small Graduate School of Information Science and Technology,\\
\small The University of Tokyo, Japan}
\date{}

\begin{document}
\maketitle

\begin{abstract}
We resolve the conjecture of Koutis and Williams that $k$-paths can be
approximately counted in $O^*(2^k)$ time.  For a simple directed graph with
$n$ vertices and $m$ arcs, our randomized algorithm returns a
$(1\pm\varepsilon)$-approximation with failure probability at most $\delta$ in
$O^*(2^k\varepsilon^{-2}\log(1/\delta))$ time.
\end{abstract}

\section{Introduction}
\label{sec:introduction}

A directed $k$-path is an ordered tuple $(v_1,\ldots,v_k)$ of distinct
vertices such that $(v_i,v_{i+1})$ is an arc for every $1\leq i<k$.  The
$k$-path problem asks whether a given graph contains such a path.  It is a
benchmark problem in parameterized algorithms~\cite{brandDellHusfeldt2018,lokshtanovSaurabhZehavi2021}:
after the early
$k!\operatorname{poly}(n)$-time algorithm of Monien and the color-coding method
of Alon, Yuster, and Zwick~\cite{alonYusterZwick1995,monien1985}, algebraic
fingerprints led to randomized algorithms with running times
$O^*(2^{3k/2})$~\cite{koutis2008} and $O^*(2^k)$ for directed
graphs~\cite{williams2009}.  Here and throughout, $O^*$ suppresses factors
polynomial in the parameter and the input size, but not the displayed
dependence on $\varepsilon$ and $\delta$.  In contrast, exactly counting
$k$-paths is $\#\mathrm W[1]$-hard for both directed and undirected
graphs~\cite{flumGrohe2004}.  This gap raises a natural intermediate question.
\begin{quote}
Can directed $k$-paths be approximately counted in $O^*(2^k)$ time?
\end{quote}

The existence of such an algorithm was conjectured by
Koutis and Williams~\cite{koutisWilliams2016Fingerprints}.
Color-coding gives a randomized $(1\pm\varepsilon)$-approximation in
$O((2e)^km\varepsilon^{-2})$ time for constant success
probability~\cite{alonEtAl2008}.
Extensor-coding subsequently gave a faster algorithm with a running time of
$O^*(4^k)$~\cite{brandDellHusfeldt2018}.
Using representative weight functions, Lokshtanov, Saurabh, and Zehavi
improved the running time to
$O((2.619^k+n^{o(1)})\varepsilon^{-2}(n+m))$~\cite{lokshtanovSaurabhZehavi2021}.
We resolve the conjecture of Koutis and
Williams:

\begin{restatable}{theorem}{directedpaththeorem}
\label[theorem]{thm:directed-paths}
Let $k\geq1$ and $0<\varepsilon,\delta<1$.  Given a simple directed graph $G$, there
is a randomized algorithm that returns a $(1\pm\varepsilon)$-approximation to
the number of directed $k$-paths in $G$ with success probability at least
$1-\delta$ and runs in $O^*(2^k\varepsilon^{-2}\log(1/\delta))$ time.
\end{restatable}

The $k$-path algorithm is a consequence of a more general theorem for counting
multilinear monomials in polynomials computed by homogeneous $0$--$1$
right-skew circuits.
Associate a directed walk $(v_1,\ldots,v_k)$ with the noncommutative monomial
$x_{v_1}\cdots x_{v_k}$.  The walk is a path exactly when this monomial is
multilinear, meaning that no variable occurs twice.
Noncommutativity preserves
the vertex order, so different walks do not collapse into the same monomial.

We call a noncommutative polynomial a \emph{$0$--$1$ polynomial} if all its
coefficients belong to $\{0,1\}$.  For such a polynomial $P$, let $N(P)$ be
the number of its multilinear monomials.  A homogeneous $0$--$1$ right-skew
circuit is built from constants $0$ and $1$, addition, and product gates that
multiply a previously computed polynomial from the right by one variable.
Every gate is required to compute a $0$--$1$ polynomial.  See
\Cref{sec:circuit-model} for the formal model.

\begin{restatable}{theorem}{rightskewtheorem}
\label{thm:right-skew}
Let $P$ be a noncommutative $0$--$1$ polynomial of degree $k\geq1$ computed
by a homogeneous $0$--$1$ right-skew circuit $C$.  For
$0<\varepsilon,\delta<1$, there is a randomized algorithm that, given $C$,
returns $\widehat N$ such that
$\Pr[|\widehat N-N(P)|\leq\varepsilon N(P)]\geq1-\delta$ and runs in
$O^*(2^k\varepsilon^{-2}\log(1/\delta))$ time.
\end{restatable}

\paragraph{Related work.}
Arvind and Raman gave early randomized parameterized counting approximation algorithms
for patterns of bounded treewidth, including
paths~\cite{arvindRaman2002}.  Alon and Gutner later used balanced families of
perfect hash functions to derandomize color-coding and obtain deterministic
multiplicative approximations for counting
paths~\cite{alonGutner2010}.

Using lifted labels in the exterior algebra on $\mathbb Z^{2k}$,
Brand, Dell, and Husfeldt gave a randomized $O^*(4^k\varepsilon^{-2})$-time
approximation for directed $k$-path~\cite{brandDellHusfeldt2018}.
Lokshtanov, Bj\"orklund, Saurabh, and Zehavi subsequently obtained
polynomial-space algorithms for directed and undirected $k$-path counting.
Their randomized running time was
$O^*(4^{k+O(\log k(\log k+\log(1/\varepsilon)))})$, whereas their
deterministic running time was
$O^*(4^{k+O(\sqrt{k}(\log^2 k+\log^2(1/\varepsilon)))})$
\cite{lokshtanovEtAl2021PolynomialSpace}.

Lokshtanov, Saurabh, and Zehavi~\cite{lokshtanovSaurabhZehavi2021} used
representative weight functions to obtain the
$O^*(2.619^k\varepsilon^{-2})$-time approximation for $k$-path.
More generally, for a homogeneous degree-$k$ commutative polynomial represented
by a monotone arithmetic circuit $C$, they gave a randomized
$(1\pm\varepsilon)$-approximation to the sum of the coefficients of its
multilinear monomials in
$O^*(3.841^k\varepsilon^{-6})$ time.
Pratt gave a randomized
constant-success $(1\pm\varepsilon)$-approximation for the corresponding
coefficient sum of a homogeneous commutative polynomial with nonnegative real
coefficients, given by black-box evaluation, in
$O^*(4.075^k\varepsilon^{-2}\log(1/\varepsilon))$ time and polynomial
space~\cite{pratt2019}.

For the decision problem, group-algebra fingerprints annihilated
nonmultilinear terms and yielded randomized multilinear detection
algorithms~\cite{koutis2008,koutisWilliams2016Fingerprints,williams2009}.
Brand and Pratt developed a symbolic-differentiation framework for skew
circuits based on determinants~\cite{brandPratt2021}.  Determinantal sieving
detected multilinear monomials whose supports formed bases of a linear matroid.
Over general fields, it evaluated a circuit over an exterior
algebra~\cite{eibenkoanaWahlstrom2025}.  Dynamic representative sets gave a
deterministic
$O^*(2^{k+O(\sqrt{k}\log^2 k)})$-time algorithm for weighted directed
$k$-path in the word-RAM model, with every edge weight stored in one
word~\cite{nederlof2025}.  For undirected $k$-path, narrow sieves gave a
polynomial-space randomized $O^*(1.66^k)$-time decision algorithm with
constant one-sided error~\cite{bjorklundHusfeldtkaskikoivisto2017}.

\paragraph{Organization.}
\Cref{sec:preliminaries} formalizes the polynomial, tensor, circuit, and
exterior-algebra notation used later.  \Cref{sec:estimator}
defines the estimator and proves its approximation guarantee.
\Cref{sec:implementation} gives its circuit implementation and proves
\Cref{thm:right-skew}.  \Cref{sec:applications}
derives \Cref{thm:directed-paths,thm:set-packing}.

\section{Preliminaries}
\label{sec:preliminaries}

\paragraph{Polynomials and tensors.}
\label{sec:circuit-model}

Fix $n\geq1$.  For a positive integer $r$, let $[r]=\{1,\ldots,r\}$.  We work
in the noncommutative polynomial ring
$\mathbb Z\langle x_1,\ldots,x_n\rangle$.  For every positive integer $i$,
each degree-$i$ monomial is uniquely indexed by a tuple
$\tau=(\tau_1,\ldots,\tau_i)\in[n]^i$, and we write
$x_\tau=x_{\tau_1}\cdots x_{\tau_i}$.  The monomial $x_\tau$ is
\emph{multilinear} if the entries of $\tau$ are pairwise distinct.

A homogeneous degree-$i$ polynomial has the form
$P=\sum_{\tau\in[n]^i}c_\tau x_\tau$.  It is a \emph{$0$--$1$ polynomial} if
$c_\tau\in\{0,1\}$ for every $\tau\in[n]^i$.  For such a polynomial, define
$N(P)=\sum_{\tau\in[n]^i:x_\tau\text{ multilinear}}c_\tau$.  The polynomial
$P$ can equivalently be represented by its coefficient tensor
$(c_\tau)_{\tau\in[n]^i}\in\mathbb Z^{[n]^i}$.  More generally, a real tensor
$A\in\mathbb R^{I_1\times\cdots\times I_r}$ is an array indexed by finite sets
$I_1,\ldots,I_r$.  Its Frobenius norm is
\[
  \|A\|_{\mathrm F}
  =\left(\sum_{(a_1,\ldots,a_r)\in I_1\times\cdots\times I_r}
    A(a_1,\ldots,a_r)^2\right)^{1/2}.
\]
To \emph{flatten} a tensor with specified row coordinates means to view the
same entries as a matrix whose row index is the Cartesian product of these
coordinates and whose column index is the Cartesian product of the remaining
coordinates.  Flattening only reindexes the entries and therefore preserves
the Frobenius norm.

\paragraph{Circuits.}
A \emph{homogeneous $0$--$1$ right-skew circuit} is a directed acyclic graph
with a designated output gate.  Each gate is a constant input gate labeled
$0$ or $1$, a binary addition gate $g=h_1+h_2$, or a product gate $g=h x_j$
with $j\in[n]$, where $h,h_1,h_2$ are predecessor gates of $g$.  These
operations recursively define the noncommutative polynomial computed at each
gate, and we require every such polynomial to be a $0$--$1$ polynomial.  Every
gate also has a syntactic degree.  Constants have degree zero, the two inputs
and the output of an addition gate have the same degree, and $h x_j$ has degree
one more than $h$.

\paragraph{Exterior algebra.}
\label{sec:exterior-algebra}

We adapt the concrete presentation of exterior algebra from Brand, Dell, and
Husfeldt~\cite[Section~2]{brandDellHusfeldt2018} to the coefficient ring
$\mathbb Z$.
Let $e_1,\ldots,e_k$ be the standard basis of $\mathbb Z^k$.  For every
$I\subseteq[k]$, let $e_I$ be a formal basis element, where
$e_\varnothing=1$ and $e_{\{j\}}=e_j$.  For $0\leq r\leq k$, the $r$th
\emph{exterior power} $\Lambda^r(\mathbb Z^k)$ is the free $\mathbb Z$-module
with basis $\{e_I:I\subseteq[k],\ |I|=r\}$.  The \emph{exterior algebra} is
$\Lambda(\mathbb Z^k)=\bigoplus_{r=0}^k\Lambda^r(\mathbb Z^k)$.
We identify $\mathbb Z^k$ with $\Lambda^1(\mathbb Z^k)$ through the basis
$e_1,\ldots,e_k$.

The \emph{wedge product} is the $\mathbb Z$-bilinear multiplication determined
on basis elements by
\[
  e_I\wedge e_J=
  \begin{cases}
    0,&I\cap J\neq\varnothing,\\
    (-1)^{|\{(a,b)\in I\times J:a>b\}|}e_{I\cup J},
      &I\cap J=\varnothing.
  \end{cases}
\]
It is associative, respects the grading
$\Lambda^r(\mathbb Z^k)\wedge\Lambda^s(\mathbb Z^k)
\subseteq\Lambda^{r+s}(\mathbb Z^k)$, and is alternating on vectors.  In
particular, $y\wedge y=0$ for every $y\in\mathbb Z^k$.  If
$I=\{i_1<\cdots<i_r\}$, then
$e_I=e_{i_1}\wedge\cdots\wedge e_{i_r}$.  Thus
$\Lambda^r(\mathbb Z^k)$ has $\binom{k}{r}$ coordinates and the full exterior
algebra has $2^k$ coordinates.  Finally, for
$y_1,\ldots,y_k\in\mathbb Z^k$, expanding the wedge product gives
\[
  y_1\wedge\cdots\wedge y_k
  =\det[y_1\ \cdots\ y_k]e_{[k]}.
\]

\section{The algorithm}

In this section, we prove \Cref{thm:right-skew}.

\rightskewtheorem*

Throughout this section, fix a polynomial $P$ and a homogeneous $0$--$1$
right-skew circuit $C$ computing it as in \Cref{thm:right-skew}, with variables
$x_1,\ldots,x_n$.  For every $\tau\in[n]^k$, let $c_\tau$ be the coefficient
of $x_\tau$ in $P$.  Let $N=N(P)$, set $\ell=2k$, and fix
$0<\varepsilon,\delta<1$.  In this section, $O^*$ suppresses factors
polynomial in $k$ and in the encoding size of $C$, but not the displayed
dependence on $\varepsilon$ and $\delta$.

Our algorithm returns $0$ whenever $N=0$, so we assume $N>0$ below.

We call a random variable $Y$ an \emph{efficient estimator} of $N$ if
$\mathbb E(Y)=N>0$ and $\mathbb E(Y^2)/(\mathbb E(Y))^2$ is bounded by
$k^{O(1)}$.  The standard median-of-means bound shows that
\[
  O\!\left(\frac{\mathbb E(Y^2)}{(\mathbb E(Y))^2}
    \varepsilon^{-2}\log(1/\delta)\right)
\]
independent samples of $Y$ suffice to compute $\widehat N$ such that
$\Pr[|\widehat N-N|>\varepsilon N]\leq\delta$.  See, for example,
Mitzenmacher and Upfal~\cite[Sections~3.3--3.4]{mitzenmacherUpfal2017}.
It therefore suffices to construct an efficient estimator of $N$ and compute
each sample in $O^*(2^k)$ time.

\subsection{Constructing an efficient estimator}
\label{sec:estimator}

Let $\sigma_{j,r}$, $j\in[n]$ and $r\in[k]$, be mutually independent random
variables, each uniform on $\{-1,1\}$, and define
$\sigma_j=(\sigma_{j,1},\ldots,\sigma_{j,k})^{\mathsf T}\in\mathbb Z^k$.
For every $\tau=(\tau_1,\ldots,\tau_k)\in[n]^k$, let
$D_\tau=\det[\sigma_{\tau_1}\ \cdots\ \sigma_{\tau_k}]$.  Define
$\Psi\in\mathbb R^{[n]^k}$ by $\Psi(\tau)=c_\tau D_\tau$.  Recall that
its Frobenius norm
is $\|\Psi\|_{\mathrm F}=(\sum_{\tau\in[n]^k}\Psi(\tau)^2)^{1/2}$.  Since
$c_\tau\in\{0,1\}$ for every $\tau\in[n]^k$, we have
\[
  \|\Psi\|_{\mathrm F}^2
  =\sum_{\tau\in[n]^k}c_\tau^2D_\tau^2
  =\sum_{\tau\in[n]^k}c_\tau D_\tau^2.
\]
If $x_\tau$ is not multilinear, then the matrix defining $D_\tau$ has two
equal columns, so $D_\tau=0$.

We first show that the normalized squared Frobenius norm
$\|\Psi\|_{\mathrm F}^2/k!$ is an efficient estimator of $N$.
The algorithm of Brand, Dell, and Husfeldt uses this estimator~\cite[Section~3.6]{brandDellHusfeldt2018}.  We include the proof in
our setting for completeness.

\begin{lemma}
\label[lemma]{lem:determinant-moments}
$\mathbb E(\|\Psi\|_{\mathrm F}^2)=k!N$ and
$\mathbb E(\|\Psi\|_{\mathrm F}^4)\leq(k!)^2k^3N^2$.
\end{lemma}

\begin{proof}
Fix $\tau\in[n]^k$ such that $x_\tau$ is multilinear.  Since
$\tau_1,\ldots,\tau_k$ are distinct, $D_\tau$ has the same distribution as
$\det(B)$, where $B$ is a $k\times k$ \emph{random sign matrix} with
independent entries uniform in $\{-1,1\}$.  The determinant-moment formulas
of Nyquist, Rice, and Riordan~\cite{nyquistRiceRiordan1954} imply
\[
  \mathbb E(\det(B)^2)
  =k!
  \qquad\text{and}\qquad
  \mathbb E(\det(B)^4)\leq(k!)^2k^3.
\]

Let $\mathcal M$ be the set of $\tau\in[n]^k$ such that $c_\tau=1$ and
$x_\tau$ is multilinear.  Then $|\mathcal M|=N$.  Since $D_\tau=0$ when
$x_\tau$ is not multilinear,
\[
  \|\Psi\|_{\mathrm F}^2=\sum_{\tau\in\mathcal M}D_\tau^2.
\]
Taking expectations gives
$\mathbb E(\|\Psi\|_{\mathrm F}^2)=k!N$.  Moreover,
Cauchy--Schwarz gives
\[
  \|\Psi\|_{\mathrm F}^4
  =\left(\sum_{\tau\in\mathcal M}D_\tau^2\right)^2
  \leq\left(\sum_{\tau\in\mathcal M}1\right)
       \left(\sum_{\tau\in\mathcal M}D_\tau^4\right)
  =N\sum_{\tau\in\mathcal M}D_\tau^4.
\]
Taking expectations and using
$\mathbb E(D_\tau^4)\leq(k!)^2k^3$ for every $\tau\in\mathcal M$ gives
$\mathbb E(\|\Psi\|_{\mathrm F}^4)\leq(k!)^2k^3N^2$.
\end{proof}

The extensor-coding construction of Brand, Dell, and
Husfeldt computes $\|\Psi\|_{\mathrm F}^2$ in $O^*(4^k)$ arithmetic
operations~\cite{brandDellHusfeldt2018}.  It uses
\emph{lifts} that map each
$\sigma_j\in\mathbb Z^k\subseteq\Lambda(\mathbb Z^k)$ to an element of
$\Lambda^2(\mathbb Z^{2k})$ and evaluates $C$ over
$\Lambda(\mathbb Z^{2k})$.  We instead compress $\Psi$ using random sign
matrices.

\paragraph{Compression by random sign matrices.}
Independently of the vectors $\sigma_1,\ldots,\sigma_n$, let
$R_1\in\{-1,1\}^{[\ell]\times[n]}$ and
$R_i\in\{-1,1\}^{[\ell]\times([\ell]\times[n])}$ for every $2\leq i\leq k$
be random matrices whose entries are mutually independent and uniform.
Flatten $\Psi$ with its first coordinate as the row coordinate,
multiply it from the left by $R_1$, and let the resulting tensor be
$S_1\in\mathbb R^{[\ell]\times[n]^{k-1}}$.  Here and below, a factor $[n]^0$
is omitted.  For each $i=2,\ldots,k$, flatten
$S_{i-1}\in\mathbb R^{[\ell]\times[n]^{k-i+1}}$ with its first two coordinates
as the row coordinates, multiply it from the left by $R_i$, and denote the
result by $S_i\in\mathbb R^{[\ell]\times[n]^{k-i}}$.  Thus
$S_k\in\mathbb R^{[\ell]}\cong\mathbb R^\ell$.  For every finite-dimensional real
vector $y$, write $\|y\|_2^2$ for the sum of the squares of its coordinates.
The final vector is given, for every $a_k\in[\ell]$, by
\begin{equation}
\label{eq:compression-coordinates}
  S_k(a_k)
  =\sum_{\tau\in[n]^k}\Psi(\tau)
   \sum_{a_1,\ldots,a_{k-1}\in[\ell]}
   R_1[a_1,\tau_1]
   \left(\prod_{t=2}^kR_t[a_t,(a_{t-1},\tau_t)]\right).
\end{equation}

We now verify that $S_k$ yields an efficient estimator of $N$.  The key point
is that, after normalization, successive compression by random sign matrices
preserves the expected squared Frobenius norm and increases its fourth moment
by only a constant factor.

\begin{lemma}
\label[lemma]{lem:estimator-guarantee}
$\|S_k\|_2^2/(\ell^k k!)$ is an efficient estimator of $N$.
\end{lemma}

\begin{proof}
We use the following one-step estimate.

\begin{claim}
Let $d$ and $t$ be positive integers.  Let $X$ be a random matrix in
$\mathbb R^{d\times t}$, and let
$R\in\{-1,1\}^{\ell\times d}$ be a random sign matrix independent of $X$.
Then
\[
  \mathbb E(\|RX\|_{\mathrm F}^2)
  =\ell\,\mathbb E(\|X\|_{\mathrm F}^2)
  \quad\text{and}\quad
  \mathbb E(\|RX\|_{\mathrm F}^4)
  \leq\ell^2(1+2/\ell)\mathbb E(\|X\|_{\mathrm F}^4).
\]
\end{claim}

\begin{proof}
First suppose that $X\in\mathbb R^{d\times t}$ is fixed.

Define $M=XX^{\mathsf T}$, and write $\operatorname{tr}M$ for the sum of its
diagonal entries.  Let
$r_1,\ldots,r_\ell\in\{-1,1\}^d$ be the rows of $R$, and define
$Q_j=r_j^{\mathsf T}Mr_j$ for every $j\in[\ell]$.  The $j$th row of $RX$
is $r_j^{\mathsf T}X$.  Therefore,
\[
\begin{aligned}
  \|RX\|_{\mathrm F}^2
  &=\sum_{j=1}^\ell\|r_j^{\mathsf T}X\|_2^2
   =\sum_{j=1}^\ell r_j^{\mathsf T}XX^{\mathsf T}r_j
   =\sum_{j=1}^\ell Q_j.
\end{aligned}
\]

Since $M$ is symmetric and the coordinates of $r_j$ have mean zero and
square one,
\[
  Q_j=\operatorname{tr}M+
  2\sum_{1\leq u<v\leq d}M_{uv}r_{j,u}r_{j,v}
  \quad\text{and}\quad
  \mathbb E(Q_j)=\operatorname{tr}M.
\]
When the centered expression is squared, every cross term indexed by two
distinct pairs has expectation zero because some coordinate occurs exactly
once.  Hence
\[
  \operatorname{Var}(Q_j)
  =\mathbb E\bigl((Q_j-\mathbb E(Q_j))^2\bigr)
  =4\mathbb E\!\left(\left(\sum_{1\leq u<v\leq d}
    M_{uv}r_{j,u}r_{j,v}\right)^2\right)
  =4\sum_{1\leq u<v\leq d}M_{uv}^2
  \leq2\|M\|_{\mathrm F}^2.
\]
Let $\lambda_1,\ldots,\lambda_d\geq0$ be the eigenvalues of the positive
semidefinite matrix $M$.  Then
\[
  \|M\|_{\mathrm F}^2=\operatorname{tr}(M^2)
  =\sum_{h=1}^d\lambda_h^2
  \leq\left(\sum_{h=1}^d\lambda_h\right)^2
  =(\operatorname{tr}M)^2,
\]
and $\operatorname{tr}M=\operatorname{tr}(XX^{\mathsf T})
=\|X\|_{\mathrm F}^2$.  Since the $Q_j$ are independent copies, linearity of
expectation and additivity of variance give
$\mathbb E(\|RX\|_{\mathrm F}^2)=\ell\operatorname{tr}M=\ell\|X\|_{\mathrm F}^2$ and
$\operatorname{Var}(\|RX\|_{\mathrm F}^2)
\leq2\ell(\operatorname{tr}M)^2$.  Therefore,
\[
  \mathbb E(\|RX\|_{\mathrm F}^4)
  =\operatorname{Var}(\|RX\|_{\mathrm F}^2)
   +(\mathbb E(\|RX\|_{\mathrm F}^2))^2
  \leq2\ell\|X\|_{\mathrm F}^4+\ell^2\|X\|_{\mathrm F}^4
  =\ell^2(1+2/\ell)\|X\|_{\mathrm F}^4.
\]

Since $X$ is a random matrix independent of $R$, conditioning on $X$ gives
$\mathbb E(\|RX\|_{\mathrm F}^2)
=\ell\,\mathbb E(\|X\|_{\mathrm F}^2)$.  It also gives
$\mathbb E(\|RX\|_{\mathrm F}^4)
\leq\ell^2(1+2/\ell)\,\mathbb E(\|X\|_{\mathrm F}^4)$.
\end{proof}

We apply the claim at each flattening.  Let $S_0=\Psi$.  For each $i\in[k]$,
the matrix multiplied by $R_i$ is a flattening of $S_{i-1}$.  It depends only on
$\sigma_1,\ldots,\sigma_n,R_1,\ldots,R_{i-1}$ and is therefore independent
of $R_i$.
Since flattening preserves the Frobenius norm, applying the preceding bounds
inductively gives, for every $i\in[k]$,
\[
\begin{aligned}
  \mathbb E(\|S_i\|_{\mathrm F}^2)
  &=\ell\,\mathbb E(\|S_{i-1}\|_{\mathrm F}^2),\\
  \mathbb E(\|S_i\|_{\mathrm F}^4)
  &\leq \ell^2(1+2/\ell)\mathbb E(\|S_{i-1}\|_{\mathrm F}^4).
\end{aligned}
\]

By iterating from $S_0=\Psi$, using
$\|S_k\|_2=\|S_k\|_{\mathrm F}$, and applying
\Cref{lem:determinant-moments}, we obtain
\[
\begin{aligned}
  \mathbb E(\|S_k\|_2^2/(\ell^k k!))
  &=\mathbb E(\|\Psi\|_{\mathrm F}^2)/k!=N,\\
  \mathbb E((\|S_k\|_2^2/(\ell^k k!))^2)
  &\leq\left(1+\frac2\ell\right)^kk^3N^2
   =\left(1+\frac1k\right)^kk^3N^2.
\end{aligned}
\]
Since $(1+1/k)^k\leq e$, the lemma follows.
\end{proof}

By \Cref{lem:estimator-guarantee} and the median-of-means bound above,
$O^*(\varepsilon^{-2}\log(1/\delta))$ independent trials suffice to return
$\widehat N$ such that
$\Pr[|\widehat N-N|>\varepsilon N]\leq\delta$.  Thus it remains only to compute
each sample within the claimed running time.

\subsection{Sampling from the efficient estimator}
\label{sec:implementation}

Use the random vectors $\sigma_1,\ldots,\sigma_n$ and random matrices
$R_1,\ldots,R_k$ defined in \Cref{sec:estimator}.  Evaluate the ancestors of
the output gate over the exterior algebra $\Lambda(\mathbb Z^k)$ from
\Cref{sec:exterior-algebra}.

For every gate $g$ of degree $i\geq1$ and every $\tau\in[n]^i$, let
$c_g(\tau)$ be the coefficient of $x_\tau$ in the noncommutative polynomial
computed at $g$; since $C$ is a $0$--$1$ circuit, $c_g(\tau)\in\{0,1\}$.
The values below store signed sums of these monomial contributions after
replacing variables by their exterior-algebra labels.  Evaluate these gates in
topological order.  Degree-zero gates store their values in $\{0,1\}$.
For every gate $g$ of degree $i\geq1$ and every $a\in[\ell]$, store
$F_{g,a}\in\Lambda^i(\mathbb Z^k)$.  Addition gates are evaluated
componentwise.  If $g=h x_j$ has degree one, we may assume that $h$ stores
$1$, since otherwise $g$ computes zero and can be ignored.  Define
$F_{g,a}=R_1[a,j]\sigma_j$.  If $g=h x_j$ has degree $i\geq2$, define
$F_{g,a}=(\sum_{b=1}^\ell R_i[a,(b,j)]F_{h,b})\wedge \sigma_j$.

For every gate $g$ of degree $i\geq1$ and every $a_i\in[\ell]$, we claim that
\begin{equation}
\label{eq:circuit-contraction}
\begin{aligned}
  F_{g,a_i}={}&
  \sum_{\tau\in[n]^i}c_g(\tau)
  \sum_{a_1,\ldots,a_{i-1}\in[\ell]}
  R_1[a_1,\tau_1]
  \left(\prod_{t=2}^iR_t[a_t,(a_{t-1},\tau_t)]\right)
  \cdot
  \sigma_{\tau_1}\wedge\cdots\wedge \sigma_{\tau_i}.
\end{aligned}
\end{equation}
For $i=1$, the inner sum is omitted and the product is interpreted as $1$.

We prove~\eqref{eq:circuit-contraction} by induction over the gates.  At a
degree-one product gate $g=h x_j$, we may assume that $h$ stores $1$.  The
coefficient of $x_j$ is then one, all other coefficients are zero, and
$F_{g,a_1}=R_1[a_1,j]\sigma_j$.  At an addition gate
$g=h_1+h_2$, both the coefficients and the stored elements are the sums of
the corresponding values at $h_1$ and $h_2$.

Now let $g=h x_j$ have degree $i\geq2$.  A monomial
$x_\tau=x_{\tau_1}\cdots x_{\tau_i}$ has nonzero coefficient at $g$ only if
$\tau_i=j$, and in this case
$c_g(\tau)=c_h((\tau_1,\ldots,\tau_{i-1}))$.  Substituting the expansion of
each $F_{h,b}$ into the recurrence makes the sum over $b$ the sum over
$a_{i-1}$, introduces the factor
$R_i[a_i,(a_{i-1},\tau_i)]$, and appends $\sigma_{\tau_i}$ to the wedge
product.
This proves~\eqref{eq:circuit-contraction}.

Let $o$ be the output gate.  Since $\Psi(\tau)=c_\tau D_\tau$,
$c_o(\tau)=c_\tau$, and
$\sigma_{\tau_1}\wedge\cdots\wedge \sigma_{\tau_k}=D_\tau e_{[k]}$,
\eqref{eq:compression-coordinates} and
\eqref{eq:circuit-contraction} give, for every $a_k\in[\ell]$,
\[
\begin{aligned}
  F_{o,a_k}
  ={}\sum_{\tau\in[n]^k}c_\tau D_\tau
  \sum_{a_1,\ldots,a_{k-1}\in[\ell]}
  R_1[a_1,\tau_1]
  \left(\prod_{t=2}^kR_t[a_t,(a_{t-1},\tau_t)]\right)e_{[k]}
  ={}S_k(a_k)e_{[k]}.
\end{aligned}
\]
Thus the stored output coefficients determine
$\|S_k\|_2^2/(\ell^k k!)
=(\ell^k k!)^{-1}\sum_{a=1}^\ell S_k(a)^2$ using only rational arithmetic.

\paragraph{Running time.}
Let $L$ be the encoding size of $C$.  After relabeling the variables occurring
in $C$, we may assume that $n\leq L$.  By \Cref{sec:exterior-algebra}, an
element of $\Lambda^i(\mathbb Z^k)$ has at most $2^k$ coefficients.  Since
$\ell=2k$, every gate can be evaluated using $2^k\operatorname{poly}(k)$
integer operations.
For a degree-$i$ gate, each coefficient in
\eqref{eq:circuit-contraction} is a sum of at most $L^i\ell^{i-1}$ terms, each
of which is a coefficient of an $i$-fold wedge of sign vectors and hence has
absolute value at most $i!$.  Thus every integer computed in one trial has
$O(k(1+\log L+\log k))$ bits.  Squaring and summing the output coefficients
changes this bound only by a constant factor.  Therefore every integer
operation within one trial takes time polynomial in $k$ and $L$, and one trial
takes $O^*(2^k)$ time.  Each estimator value has denominator $\ell^k k!$, and
averaging $O(k^3\varepsilon^{-2})$ values adds only
$O(1+\log k+\log(1/\varepsilon))$ bits to its numerator and denominator.  The
group averages have a common denominator, so their median can be selected
using a linear number of comparisons between their numerators.  These
comparisons and the remaining amplification arithmetic take polynomial time.
The amplification in \Cref{sec:estimator} uses
$O^*(\varepsilon^{-2}\log(1/\delta))$ trials, giving a total running time of
$O^*(2^k\varepsilon^{-2}\log(1/\delta))$.

\section{Applications}
\label{sec:applications}

For completeness, we give the noncommutative-polynomial construction
underlying \Cref{thm:directed-paths}.  We also derive, by the same approach, a
result for set packing.

\subsection{Directed paths}
\label{sec:directed-paths}

A directed $k$-path is an ordered tuple $(v_1,\ldots,v_k)$ of distinct
vertices such that $(v_i,v_{i+1})$ is an arc for every $1\leq i<k$.

\directedpaththeorem*

\begin{proof}
Let the input digraph have vertex set $[n]$ and arc set $A$, where $|A|=m$.
If $n<k$, return $0$, so assume $n\geq k$.  Define $Z_0=0$ and
$Z_i=Z_{i-1}x_1$ for $i\in[k]$; thus $Z_i$ is a syntactically homogeneous
zero polynomial of degree $i$.
For every $v\in[n]$, define $P_{v,1}=x_v$.  For $i=2,\ldots,k$ and
$v\in[n]$, define
\[
  P_{v,i}=\left(\sum_{u:(u,v)\in A}P_{u,i-1}\right)x_v,
\]
where an empty sum is represented by $Z_{i-1}$.  The polynomial $P_{v,i}$
stores the directed walks on $i$ vertices that end at $v$.

We prove this claim by induction on $i$.  It holds for $i=1$.  For $i\geq2$,
each summand indexed by an arc $(u,v)$ appends $v$ to a walk ending at $u$.
Conversely, deleting the final vertex of a walk ending at $v$ identifies its
summand.  The monomial records the entire vertex sequence and therefore occurs
with coefficient one.

Let $P_k=\sum_{v\in[n]}P_{v,k}$.  The last variable identifies the endpoint,
so every monomial of $P_k$ still has coefficient one.  Its monomials encode
all directed walks on $k$ vertices, and a monomial is multilinear exactly when
its variables are distinct.  Thus the multilinear monomials of $P_k$ are in
bijection with the directed $k$-paths.

Every intermediate polynomial is $0$--$1$: the predecessor sums combine
disjoint supports, right multiplication is injective on monomials, and the
final sum combines monomials with distinct last variables.  Thus these
recurrences and the final sum form a polynomial-size homogeneous $0$--$1$
right-skew circuit.  Applying \Cref{thm:right-skew} to $P_k$ gives the claimed
approximation and running time.
\end{proof}

\subsection{Set packing}
\label{sec:set-packing}

A $k$-packing is a $k$-element subfamily of pairwise disjoint sets.  Counting
$k$-matchings is $\#\mathrm W[1]$-hard even in bipartite
graphs~\cite{curticapeanMarx2014}.
Liu, Wang, and Wang gave fixed-parameter randomized approximation schemes for
counting packings in families of sets of fixed size~\cite{liuWangWang2018}.
For families of $d$-element sets, the representative-weight method gives a
randomized $(1\pm\varepsilon)$-approximation in
$O^*(2.619^{dk}\varepsilon^{-2})$ time~\cite{lokshtanovSaurabhZehavi2021}.
The following theorem improves this exponential factor to $2^{dk}$.

\begin{theorem}
\label[theorem]{thm:set-packing}
Let $d,k\geq1$ and $0<\varepsilon,\delta<1$.  Given a family $\mathcal E$ of
$m\geq k$ distinct $d$-element sets, there is a randomized algorithm
that returns a $(1\pm\varepsilon)$-approximation to the number of $k$-packings
in $\mathcal E$ with success probability at least $1-\delta$ and runs in
$O^*(2^{dk}\varepsilon^{-2}\log(1/\delta))$ time.
\end{theorem}

\begin{proof}
Let $\mathcal E=\{E_1,\ldots,E_m\}$ be the input family.
Index a variable $x_u$ by every ground element $u$.  For each $h\in[m]$,
choose an ordering $u_{h,1},\ldots,u_{h,d}$ of $E_h$ and define
$M_h=x_{u_{h,1}}\cdots x_{u_{h,d}}$.  We list selected sets in increasing
input order because a commutative product would merge distinct subfamilies
with the same union.

For every $h\in[m]$ and $0\leq j\leq\min\{h,k\}$, the polynomial $P_{h,j}$
stores the $j$-element subfamilies of $\{E_1,\ldots,E_h\}$.  Define
$P_{0,0}=1$ and $P_{h,0}=1$ for $h\in[m]$.  For every $h\in[m]$ and
$1\leq j\leq\min\{h,k\}$, define
\[
  P_{h,j}=
  \begin{cases}
    P_{h-1,j}+P_{h-1,j-1}M_h,&j<h,\\
    P_{h-1,h-1}M_h,&j=h.
  \end{cases}
\]
Expanding each multiplication by $M_h$ into $d$ consecutive right-product
gates implements these recurrences by a polynomial-size homogeneous
right-skew circuit with output $P_{m,k}$.

Induction on $h$, using $P_{h,0}=1$ as the base state, shows that for every
$1\leq j\leq\min\{h,k\}$, the monomials of $P_{h,j}$ are exactly
$M_{i_1}\cdots M_{i_j}$ indexed by the sequences
$1\leq i_1<\cdots<i_j\leq h$.  The two terms of the recurrence correspond to
subfamilies that omit or contain $E_h$, respectively.  These cases are
disjoint and exhaustive.  Each such monomial has a unique decomposition into
blocks of length $d$, and the input sets are distinct, so different
subfamilies give different monomials.  Each monomial has coefficient one.
Each intermediate right-product gate appends a fixed prefix of $M_h$ and
therefore also computes a $0$--$1$ polynomial.  Hence the circuit above is a
homogeneous $0$--$1$ right-skew circuit.

Such a monomial is multilinear exactly when the selected sets are pairwise
disjoint.  Thus $P_{m,k}$ is a homogeneous $0$--$1$ polynomial of degree $dk$,
and $N(P_{m,k})$ is the number of $k$-packings in $\mathcal E$.  Applying
\Cref{thm:right-skew} gives the claimed approximation and running time.
\end{proof}

\section*{Acknowledgments}

This work was supported in part by JST CREST Grant Number JPMJCR24Q2 and JST
ERATO Grant Number JPMJER2301.

\section*{Declaration of AI use}

The original algorithmic ideas underlying this work were proposed by OpenAI's
GPT-5.6 Pro.  The author subsequently developed and substantially simplified
the proofs, carefully checked all mathematical arguments, and takes full
responsibility for the mathematical content of the paper.

\end{document}